\documentclass{article}

\usepackage{pdfsync}

\usepackage{amssymb}
\usepackage{amsthm}

\newcommand{\calA}{\mathcal{A}}

\newcommand{\frakA}{\mathfrak{A}}
\newcommand{\frakB}{\mathfrak{B}}

\newcommand{\A}{{\mathcal A}}

\newcommand{\Nat}{\mathbb{N}}

\newcommand{\isom}{\cong}

\newcommand{\liff}{\leftrightarrow}
\newcommand{\limp}{\rightarrow}

\newcommand{\GF}{\ensuremath{\mathsf{GF}}} 
\newcommand{\muGF}{\ensuremath{\mu\mathsf{GF}}}

\newcommand{\gbisim}{\sim_\mathrm{g}}

\newcommand{\tup}[1]{{\bar {#1}}}

\theoremstyle{plain}

\newtheorem{theorem}{Theorem}
\newtheorem{proposition}[theorem]{Proposition}

\newtheorem{lemma}[theorem]{Lemma}
\newtheorem{claim}[theorem]{Claim}

\begin{document}

\author{Vince B\'ar\'any and Miko\l{}aj Boja\'nczyk \\[1em]
   \texttt{\{vbarany,bojan\}@mimuw.edu.pl}\\
   Faculty of Mathematics, Informatics and Mechanics \\ 
   University of Warsaw, Banacha 2, 02-097 Warsaw, Poland
   \thanks{Authors were supported by ERC Starting Grant ``Sosna''.}}

\title{Finite Satisfiability for Guarded Fixpoint Logic}

\date{9 February 2012}

\maketitle

\begin{abstract}
The finite satisfiability problem for guarded fixpoint logic is decidable 
and complete for 2ExpTime (resp.~ExpTime for formulas of bounded width). \\
{\it Keywords} guarded fragment, guarded fixpoint logic, finite satisfiability
\end{abstract}


\section{Introduction} \label{sec_intro}


The \emph{guarded fragment} (\GF) is a robustly decidable syntactic fragment 
of first-order logic possessing many favourable model theoretic traits, 
such as the finite model property~\cite{Gr99JSL}. The guarded fragment 
has received much attention since its conception thirteen years ago~\cite{ABN98JPL} 
and has since seen a number of variants and extensions adopted in diverse 
fields of computer science. 
One of the most powerful extensions to date, \emph{guarded fixpoint logic} 
(\muGF) was introduced by Gr\"adel and Walukiewicz in~\cite{GW99lics}, 
who showed that the satisfiability problem of guarded fixpoint logic is 
computationally no more complex than for the guarded fragment: 
{\sc 2ExpTime}-complete in general and {\sc ExpTime}-complete for formulas 
of bounded width. Guarded fixpoint logic extends the modal $\mu$-calculus 
with backward modalities, hence it does \emph{not} have the finite model property. 
Therefore, there is a finite satisfiability decision problem: to determine 
whether a formula has a finite model. Gr\"adel and Walukiewicz left the 
decidability of this problem open. Here we claim this inheritance.

\theoremstyle{plain} \newtheorem{maintheorem}[theorem]{Main Theorem}
\begin{maintheorem} \label{thrm_main}
It is decidable whether or not a given guarded fixpoint sentence is finitely 
satisfiable. The problem is {\sc 2ExpTime}-complete in general, 
and {\sc ExpTime}-complete for formulas of bounded width.
\end{maintheorem} 

As noted above the stated hardness results already hold for the guarded 
fragment~\cite{Gr99JSL}. The proof of the upper bounds combines three 
ingredients:
\begin{enumerate} \setlength{\itemsep}{-0.2em} 
  \item[i.] the tight connection between \muGF\ and 
        alternating automata~\cite{GW99lics};
  \item[ii.] decidability of emptiness of alternating automata over 
        finite graphs~\cite{Boj02};
  \item[iii.] a recent development in the finite model theory of 
        guarded logics~\cite{BGO10lics}.
\end{enumerate}
In what follows, no intricate knowledge of either~\cite{Boj02} 
or~\cite{BGO10lics} is required, the results of these papers 
are used as black boxes: 
i. \& ii. provide the algorithm and the construction of iii. proves 
its correctness. The stated time complexity results from combining 
those of i.~(Theorem~\ref{thrm_GFPtoATWA} below) 
and ii.~(Theorem~\ref{thrm_finaccept}).

\paragraph*{\bf Outline of the paper} 
Guarded fixpoint logic and related notions are introduced in 
Section~\ref{sec_GFP}. In Section~\ref{sec_automata} we define 
alternating automata on undirected graphs, and state the result 
of~\cite{Boj02}. Section~\ref{sec_tabloids} establishes the connection 
between guarded fixpoint logic and alternating automata along the lines 
of~\cite{GW99lics}. In Section~\ref{sec_alg}, we present the algorithm 
and prove its correctness using~\cite{BGO10lics}.


\section{Guarded Fixpoint Logic} \label{sec_GFP}


The guarded fragment of first-order logic comprises only formulas with a 
restricted pattern of ``guarded quantification'' and otherwise inherits 
the semantics of first-order logic. Guarded quantification takes the form 
$$
  \exists \tup{y} \left( R(\tup{x}\tup{y}) \land \varphi(\tup{x}\tup{y}) \right) 
  \quad \mathrm{ or, dually, } \quad 
  \forall \tup{y} \left( R(\tup{x}\tup{y}) \limp \varphi(\tup{x}\tup{y}) \right)
$$ 
where $R(\tup{x}\tup{y})$ is a positive literal acting as a \emph{guard} 
by effectively restricting the variables $\tup{x}$ to range only over those 
tuples occurring in the appropriate positions in the atomic relation $R$. 
Here it is meant that $\tup{x}\tup{y}$ include all free variables of $\varphi$ 
in no particular order. 
A \emph{guarded set} of elements of a relational structure $\frakA$ is 
a set whose members occur among the components of a single relational atom 
$R(\tup{a})$ of $\frakA$.
Guarded quantification can be understood as a generalisation of polyadic 
modalities of modal logic. Indeed, the guarded fragment was conceived 
precisely with this analogy in mind \cite{ABN98JPL}, therefore it is no 
coincidence that the model theory of the guarded fragment bears such 
a strong resemblance to that of modal logic \cite{GHO02}.

Guarded fixpoint logic is obtained by extending the guarded fragment of 
first-order logic with least and greatest fixpoint constructs. 
Its syntax can be defined by the following scheme  
\[\begin{array}{rcl}
  \varphi & ::= & R(\tup{x}) \ \mid \ \varphi \land \varphi' \ \mid \ \lnot \varphi 
        \ \mid \ \exists \tup{y} \left(\, R(\tup{x}\tup{y}) \land \varphi''(\tup{x}\tup{y}) \,\right) \ \mid \\[0.5em]
          &     & Z(\tup{z}) \ \mid \ [\,\mathsf{LFP}\, Z, \tup{z}\,.\, \varphi'''(Z,\tup{z})\,] (\tup{x})
                             \ \mid \ [\,\mathsf{GFP}\, Z, \tup{z}\,.\, \varphi'''(Z,\tup{z})\,] (\tup{x})
\end{array}\]
where $R$ is an arbitrary atomic relation symbol, $Z$ is a second-order 
fixpoint variable, where all free first-order variables of $\varphi''(\tup{x}\tup{y})$ 
and $\varphi'''(Z,\tup{z})$ are among those indicated, and $\varphi'''(Z,\tup{z})$ 
is required to be positive in $Z$. 
The semantics is standard: the least (or greatest) fixpoint of a formula 
$\varphi'''(Z,\tup{z})$ on a given structure is the wrt. set inclusion 
least (resp. greatest) relation $S$ satisfying $S(\tup{a}) \liff \varphi'''(S,\tup{a})$ 
for all $\tup{a}$ on the structure.
Crucially, fixpoint variables and fixpoint formulas are not allowed to 
stand as guard in a guarded quantification, only atomic relation symbols 
may act as guards. Furthermore, within sentences it can be assumed wlog. 
that in the matrix $\varphi'''(Z,\tup{z})$ of a fixpoint formula the tuple 
of free variables $\tup{z}$ is explicitly guarded \cite{GW99lics}.

Guarded fixpoint logic naturally extends the modal $\mu$-calculus with 
backward modalities. As such it can axiomatise (the necessarily infinite) 
well-founded directed acyclic graphs having no sink nodes, e.g. as follows.
\[
   \exists xy\, E(x,y) \ \land \ 
   \forall xy\, \bigg( E(x,y) \limp [\,\mathsf{LFP}\,Z,z\,.\, \forall v E(v,z) \limp Z(v)\,](x) 
                      \ \land \ \exists w E(y,w) \bigg)
\]
\paragraph*{Guarded bisimulation}
Guarded logics possess a very appealing model theory in which 
guarded bisimulation plays a similarly central role as does bisimulation 
for modal logics. 
A \emph{guarded bisimulation} \cite{ABN98JPL,GHO02} between two 
structures $\frakA_0$ and $\frakA_1$ of the same relational signature 
is a family $Z$ of partial isomorphisms $\alpha : A_0 \to A_1$ 
with $A_i \subseteq \frakA_i$, satisfying the following back-and-forth conditions. 
(i) For every $\alpha: A_0 \to A_1$ in $Z$ and every \emph{guarded} subset $B_{0}$ 
   of $\frakA_0$ there is a partial isomorphism $\gamma : C_0 \to C_1$ in $Z$ 
   with $B_0 \subseteq C_0$ and $\alpha|_{A_0 \cap C_0}=\gamma|_{A_0 \cap C_0}$. 
(ii) For every $\alpha: A_0 \to A_1$ in $Z$ and every \emph{guarded} subset $B_{1}$ 
   of $\frakA_1$ there is a partial isomorphism $\gamma : C_0 \to C_1$ in $Z$ 
   with $B_1 \subseteq C_1$ and $\alpha^{-1}|_{A_1 \cap C_1}=\gamma^{-1}|_{A_1 \cap C_1}$. 
%
We write $\frakA_0,\tup{a} \gbisim \frakA_1,\tup{b}$ to signify that 
there is a guarded bisimulation $Z$ between $\frakA_0$ and $\frakA_1$
with $(\tup{a} \mapsto \tup{b}) \in Z$ and say that $\tup{a}$ of $\frakA_0$ 
and $\tup{b}$ of $\frakA_1$ are \emph{guarded bisimilar}. 

%

Guarded bisimilarity is an equivalence relation on the set of guarded tuples 
of any relational structure, and guarded fixpoint formulas are invariant 
under guarded bisimulation \cite{GHO02}: if $\frakA,\tup{a} \gbisim \frakB,\tup{b}$ 
then for every guarded fixpoint formula $\varphi$ it holds that 
$\frakA \models \varphi(\tup{a})$ iff $\frakB \models \varphi(\tup{b})$.
%
The guarded fragment has been characterised as the guarded-bisimulation-invariant 
fragment of first-order logic, most recently even in the context of finite 
structures \cite{Otto10lics}. Similarly, guarded fixpoint logic is 
characterised as the guarded-bisimulation-invariant fragment of guarded 
second-order logic \cite{GHO02}.


\section{Alternating two-way automata} 
\label{sec_automata}



In this section, we introduce alternating automata on undirected graphs.
A similar model, namely alternating two-way automata on infinite trees, 
was used by Gr\"adel and Walukiewicz \cite{GW99lics} in their decision procedure 
for satisfiability of guarded fixpoint logic. They reduced satisfiabilty to 
the emptiness problem for alternating two-way automata on infinite trees. 
The latter problem was shown to be decidable by Vardi~\cite{Var98}.

In~\cite{Var98,Boj02,Boj04diss} a two-way automaton navigating an infinite tree 
has the choice of moving its head either to the parent or to a child node, 
or staying in its current location. 
In this paper, instead of automata on directed trees, we consider automata 
on undirected graphs. In an undirected graph, the automaton can only choose 
to stay in place or to move to a neighboring vertex. 
This is in the spirit of~\cite{GW99lics}, where automata on directed 
trees were employed, which did not actually distinguish between parent 
and child nodes.


\newcommand{\stay}{\mathsf{stay}}
\newcommand{\move}{\mathsf{move}}
\medskip
An \emph{alternating automaton on undirected graphs} is defined by: 
an input alphabet $\Sigma$, a set of states $Q$, 
a partition $Q = Q_\forall \cup Q_\exists$, an initial state $q_I$, 
a ranking function $\Omega : Q \to \Nat$ for the parity acceptance condition, 
and a transition relation
\[
	\delta \subseteq Q \times \Sigma \times \{\stay,\move\} \times Q \ .
\]
An input to the automaton is an undirected graph whose nodes are 
labelled by $\Sigma$, and a designated node $v_0$ of the graph. 
The automaton accepts an input graph $G$ from an initial node $v_0$ 
if player $\exists$ wins the parity game defined below.

The arena of the parity game consists of pairs of the form $(v,q)$, 
where $v$ is a node of $G$, and $q$ is a state of the automaton. 
The initial position in the arena is $(v_0,q_I)$. 
The rank of a position $(v,q)$, as used by the parity condition, is $\Omega(q)$. 
Let $u$ be a node of the input graph, and let $a \in \Sigma$ be its label. 
In the arena of the game, there is an edge from  $(u,q)$ to $(w,p)$ if:
\begin{list}{$\bullet$}{
  \setlength{\itemsep}{0.4\itemsep} \setlength{\topsep}{\itemsep}\setlength{\parsep}{0pt}}
  \item there is a transition $(q,a,\stay,p)$ and $u=w$; or
  \item there is a transition $(q,a,\move,p)$ and 
        $\{u,w\} \in E(G)$.
\end{list}

Some alternating automata on undirected graphs accept only infinite graphs.
 (Given a 3-coloring of a graph by $\{0,1,2\}$, edges can be directed so 
  that `target color' $-$ `source color' $\equiv$ $1$ mod~$3$.
  An automaton can verify 3-coloring and well-foundedness of the  
  induced digraph and check for an infinite forward path.)
Therefore, it makes sense to ask: does a given automaton accept some finite graph?  
This problem was shown decidable in~\cite{Boj02,Boj04diss}.

\begin{theorem}[{\cite{Boj02,Boj04diss}}] \label{thrm_finaccept}
Given a alternating  automaton on undirected graphs 
it is decidable in exponential time in the number of states of the automaton, 
whether or not it accepts some finite graph.
\end{theorem}

Formally, \cite{Boj02,Boj04diss} considered two-way automata on directed graphs 
with the automaton having transitions corresponding to: staying in the same node, 
moving forward along an edge, and moving backward along an edge. 
Clearly, the two-way model is more general than the one for undirected graphs.

\newcommand{\nodes}{\mathrm{nodes}}

\paragraph*{Undirected bisimulation}
We write $\nodes(G)$ for the nodes of a graph $G$. 
Consider two undirected graphs $G_0$ and $G_1$, with node labels. 
An undirected bisimulation is a set
\begin{eqnarray*}
 Z	\subseteq \nodes(G_0) \times \nodes(G_1)
\end{eqnarray*}
with the following properties. If $(v_0,v_1)$ belongs to $Z$, 
then the node labels of $v_0$ and $v_1$ are the same. 
Also, for any $i \in \{0,1\}$ and node $w_i$ connected to $v_i$ by an edge, 
there exits a node $w_{1-i}$ connected to $v_{1-i}$ by an edge and such that
$(w_0,w_1) \in Z$. 
We say that node $v_0$ of a graph $G_0$ is bisimilar to node $v_1$ in a graph $G_1$ 
if there is an undirected bisimulation that contains the pair $(v_0,v_1)$. 
In this case, for every alternating automaton on undirected graphs,
the automaton accepts $G_0$ from $v_0$ if and only if it accepts $G_1$ from $v_1$.

\paragraph*{Undirected unraveling}
Consider an undirected graph $G$ and $v$ a node of $G$. 
The undirected unraveling of $G$ from $v$ is the graph $T$, 
whose nodes are paths in $G$ that begin in $v$, and edges are 
placed between a path and the same path without the last node. 
The undirected unraveling is a tree. We write
\begin{eqnarray*}
	\pi : \nodes(T)\to \nodes(G)
\end{eqnarray*}
for the function that maps a path to its terminal node. 
If $G$ has node labels, then one labels the nodes of $T$ according to 
their images under $\pi$. Then, the graph of $\pi$ is an undirected 
bisimulation between $T$ and $G$.


\section{Tabloids} 
\label{sec_tabloids}


Below we work with undirected graphs representing templates of relational 
structures. We call them tabloids alluding to their semblance to the 
tableaux of \cite{GW99lics}. Tabloids are also reminiscent of the 
`guarded bisimulation invariants' of~\cite{BGO10lics}. Intuitively, 
vertices of a tabloid represent templates for guarded substructures and 
edges signify their overlap. The precise manner of overlap is implicitly 
coded  by repeated use of constant names appearing in vertex labels. 
By contrast, \cite{BGO10lics,GHO02} code overlaps explicitly as edge labels.


\paragraph*{Tabloid} 
Fix a relational signature $\Sigma$ and a set $K$ of constant names. 
A \emph{tabloid} over signature $\Sigma$ and constants $K$ is an 
undirected graph, where every node $v$ is equipped with two labels: 
  a set  $K_v \subseteq K$, called the \emph{constants of $v$}, and 
  an atomic $\Sigma$-type $\tau_v$ over $K_v$, called the \emph{type of $v$}.
If nodes $v$ and $w$ are connected by an edge in the graph, 
then the types $\tau_v$ and $\tau_w$ should agree over 
the constants from $K_v \cap K_w$. 

\paragraph*{A structure from a tree tabloid} 
Consider a tabloid $T$ whose underlying graph is a tree. 
We define a $\Sigma$-structure $\frakA(T)$ as follows. 
The universe of $\frakA(T)$ is built using pairs $(v,c)$, 
where $v$ is a vertex of $T$ and $c$ is a constant of $v$. 
The universe consists not of these pairs, but of their 
equivalence classes under the following equivalence relation: 
 $(v,c)$ and $(v',c')$ are equivalent if $c=c'$ and $c$ occurs in the label 
 of every node on the undirected path connecting $v$ and $v'$ in $T$. 
The path is unique, because the underlying graph is a tree. 
We write $[v,c]$ for an equivalence class of such a pair. 
A tuple $([v_1,c_1],\ldots,[v_n,c_n])$ satisfies a relation $R \in \Sigma$ 
in $\frakA(T)$ if there is some node $v$ such that 
\begin{equation} \label{eq:common-v}
	[v,c_1]=[v_1,c_1],\ldots,[v,c_n]=[v_n,c_n] 
\end{equation}
and $ R(c_1,\ldots,c_n)$ is implied by $\tau_v$. 
Because $T$ is a tree, this definition does not depend on the choice of $v$, 
since the set of nodes $v$ satisfying~(\ref{eq:common-v}) is connected.
It is, however, unclear how to extend this construction to cyclic tabloids. 

\paragraph*{Labelling with a formula} 
Consider a tree tabloid $T$ over constants $K$ and signature $\Sigma$. 
Let $\varphi$ be a formula over $\Sigma$. Consider a node $v$ of $T$ 
with constants $K_v$, a subformula $\psi$ of $\varphi$, and a function 
$\eta$ that maps free variables of $\psi$ to constants in $K_v$.
For $v$ and $\eta$, define a valuation $[\eta]_v$, which maps free 
variables of $\psi$ to elements of the structure $\frakA(T)$, by setting 
$ 
    [\eta]_v(x)=[v,\eta(x)] \,.
$

The \emph{$\varphi$-type} of the node $v$ is the set of pairs $(\psi,\eta)$ 
such that $\psi$ is a subformula of $\varphi$ or a literal in the signature 
of $\varphi$, and such that $\psi$ is valid in $\frakA(T)$ under the 
valuation $[\eta]_v$.
Thus each $\varphi$-type determines a unique atomic type. 
The set of  $\varphi$-types is finite and depends on $K$ and $\varphi$ 
alone, call this set $\Gamma_{\varphi,K}$.  
Given a tree tabloid $T$ and $\varphi$, we define $T_\varphi$ to be 
the tree with the same nodes and edges as $T$, but where every node 
is labelled by its $\varphi$-type.

Recall that the width of a formula is the maximal number of free variables 
in any of its subformulas. The following was established in~\cite{GW99lics}.

\begin{theorem}[\cite{GW99lics}] \label{thrm_GFPtoATWA}
Let $\varphi$ be a  guarded fixpoint sentence of width $n$ and let $K$ 
be a set of $2n$ constants. One can compute an alternating automaton 
$\calA_\varphi$ on $\Gamma_{\varphi,K}$-labelled undirected graphs, 
such that $\calA_\varphi$ accepts a tree $\Upsilon$ if and only~if 
\begin{eqnarray*}
  \Upsilon \mbox{ is of the form } T_\varphi \mbox{ for a tree tabloid } T 
      \mbox{ such that } \frakA(T) \models\varphi \ . 
\end{eqnarray*}
The number of states of $\calA_\varphi$, and the time to compute it, are 
$ 
	O(|\varphi| \cdot \exp(n)) \,.
$ 
\end{theorem}


\section{Algorithm for finite satisfiability}
\label{sec_alg}


We now propose the algorithm for finite satisfiability of 
guarded fixpoint logic. Given a formula $\varphi$, we compute the 
automaton $\A_\varphi$ using Theorem~\ref{thrm_GFPtoATWA}. 
Then, we test if the automaton $\A_\varphi$ accepts some finite graph, 
using Theorem~\ref{thrm_finaccept}. The combined running time 
clearly meets the claim of Theorem~\ref{thrm_main}. 
This section is devoted to proving the correctness of this procedure. 

\begin{proposition} \label{prop_finite}
A formula $\varphi$ of guarded fixpoint logic has a finite model if, and 
only if, the associated automaton $\A_\varphi$ accepts a finite graph.
\end{proposition}


\subsection{From a finite accepted graph to a finite model} 
\label{subsection_fingraphfinmod}

First we prove that if the automaton $\calA_\varphi$ accepts a finite 
graph $G_\varphi$, then $\varphi$ is satisfied in some finite structure.
By Theorem~\ref{thrm_GFPtoATWA}, the undirected unravelling 
of $G_\varphi$, equally accepted by $\calA_\varphi$, takes the form $T_\varphi$ 
for a tree tabloid $T$ such that $\frakA(T) \models \varphi$. 
In fact, $T$ is the undirected unravelling of the finite tabloid $G$
obtained from $G_\varphi$ by restricting its labels to atomic types. 

\begin{lemma} \label{lemma_unravelling_gbisim} 
Let $G$ be a finite tabloid and $T$ its undirected unraveling. 
Then $\gbisim$ has finite index on the set of guarded tuples of $\frakA(T)$. 
\end{lemma}
\begin{proof} 
All guarded subsets of $\frakA(T)$ are 
of the form $\{[v,c_1],\ldots,[v,c_r]\}$ where $c_1,\ldots,c_r \in K$ 
are constant names appearing in the label of $v \in \nodes(T)$. 
Let $\pi: nodes(T) \to nodes(G)$ be the natural projection from $T$ onto $G$. 
Then $(T,v) \isom (T,w)$ whenever $\pi(v)=\pi(w)$, so it suffices to show the following.
\begin{claim} \label{claim_unravelling_gbisim}
  \quad 
  $\frakA(T),([v,c_1],\ldots,[v,c_r]) \ \gbisim \ \frakA(T),([w,c_1],\ldots,[w,c_r])$ \\
  for every $v$ and $w$ such that $(T,v) \isom (T,w)$ 
  and $\{c_1,\ldots,c_r\} = K_v = K_w$. 
\end{claim}
Let for each $v$ and $w$ as in the claim $\alpha_{v,w}$ be the 
partial function mapping $[v,c] \mapsto [w,c]$ for all $c \in K_v$. 
By definition of $\frakA(T)$ we have that each $\alpha_{v,w}$ 
is a partial isomorphism among guarded subsets of~$\frakA(T)$. 
We will show that 
\begin{eqnarray*} \label{eq_alpha_vw}
    Z = \{ \, \alpha_{v,w} \mid (T,v) \isom (T,w) \, \} 
\end{eqnarray*}
is a guarded bisimulation. 
Take any $\alpha_{v,w} \in Z$ and guarded subset $B$ of $\frakA(T)$. 
Then $B = \{[u,d_1],\ldots,[u,d_s]\}$ for some $u \in \nodes(T)$ 
and constant names $D = \{d_1,\ldots,d_s\} \subseteq K_u$. 
Because $(T,v) \isom (T,w)$ there is a $y \in \nodes(T)$ 
such that $(T,v,u) \isom (T,w,y)$. 
In particular, $B \subseteq \mathrm{dom}(\alpha_{u,y})$, and 
the paths connecting $v$ with $u$ and $w$ with $y$ are isomorphic. 
We thus have for every $i \leq r$ and $j\leq s$ 
that $[v,c_i]=[u,d_j]$ iff $c_i = d_j$ and $c_i \in K_z$ for every 
node $z$ on the path connecting $v$ and $u$ (equivalently, on the path 
connecting $w$ and $y$) iff $[w,c_i]=[y,d_j]$ . 
Therefore, $\alpha_{u,y}$ and $\alpha_{v,w}$ agree on 
$\mathrm{dom}(\alpha_{u,y}) \cap \mathrm{dom}(\alpha_{v,w})$, 
and $\alpha_{u,y}^{-1}$ and $\alpha_{v,w}^{-1}$ agree 
on $\mathrm{rng}(\alpha_{u,y}) \cap \mathrm{rng}(\alpha_{v,w})$.
This shows that $Z$ satisfies the `forth property' and, by symmery, 
also the `back property', as needed. 
\end{proof}

Note that, in stark contrast to bisimulation on graphs, there is 
no apparent way of defining a quotient $\frakA(T)/{\gbisim}$. 
Nevertheless, we can obtain a finite structure guarded bisimilar 
to $\frakA(T)$ using the following result.

\begin{theorem}[{\cite[Theorem 6]{BGO10lics}, cf.~also~\cite{Otto10lics}}] \label{thrm_inversion} 
  Every relational structure on which $\gbisim$ has finite index 
  is guarded bisimilar to a finite structure.
\end{theorem}

\subsection{From a finite model to a finite accepted graph} 
\label{subsec_finmodfingraph}

Next we prove that if $\varphi$ has a finite model then $\calA_\varphi$ 
of Theorem~\ref{thrm_GFPtoATWA} accepts some finite graph. 
Recall that all graphs accepted by $\calA_\varphi$ are labelled by $\varphi$-types 
from $\Gamma_{\varphi,K}$, where $K$ is a set of $2n$ constants, with $n$ 
the width of $\varphi$.
So let $\frakA$ be a finite model of $\varphi$. Wlog.~all guarded subsets 
of $\frakA$ are of size at most $n$ (as $\varphi$ is oblivious to  
relational atoms with more than $n$ distinct components, these can be 
safely removed from $\frakA$). 

We define a finite tabloid $G$ as follows. 
Vertices of $G$ are injections $\chi : A \to K$, where $A$ is a guarded 
subset of $\frakA$. For each vertex $\chi$ its set of constants 
is $K_\chi = \mathrm{rng}(\chi)$, and its type $\tau_\chi$ is 
the image of the atomic type of $A$ in $\frakA$ under $\chi$.  
Two vertices $\chi$ and $\chi'$ are adjacent in $G$ 
just if $\chi \cup \chi'$ is an injective function. 
This ensures that adjacent nodes are labelled with consistent types,
i.e.~that $G$ is indeed a tabloid. 

Let $T$ be the undirected unraveling of $G$, and $\pi: \nodes(T) \to \nodes(G)$
the natural projection. Then $(T,v) \isom (T,w)$ whenever $\pi(v)=\pi(w)$.
From Claim~\ref{claim_unravelling_gbisim} and the guarded bisimulation invariance 
of $\muGF$ it follows that $v$ and $w$ have the same label in $T_\varphi$ 
whenever $\pi(v)=\pi(w)$.
Hence, it make sense to define $G_\varphi$ as having the same underlying 
graph as $G$ with each $\chi \in \nodes(G)$ labelled exactly as 
any and all nodes in $\pi^{-1}(\chi)$. Then $T_\varphi$ is isomorphic 
to the undirected unravelling of $G_\varphi$. 
By Theorem~\ref{thrm_GFPtoATWA}, $\calA_\varphi$ accepts $G_\varphi$ 
iff it accepts $T_\varphi$ iff $\frakA(T) \models \varphi$.
Thus, to conclude, it suffices to prove the following. 
\begin{claim}
    $ \frakA \gbisim \frakA(T) $
\end{claim}

\begin{proof}
For each $v \in \nodes(T)$, $\pi(v)$ is an injection 
$\chi_v: A_v \to K_v$ from a guarded subset $A_v$ of $\frakA$ to the set $K_v$ 
of constant names in the label of $v$. 
Let $\gamma_v: K_v \to \frakA(T)$ map each $c \in K_v$ to $[v,c]$.
Then $\gamma_v \circ \chi_v$ is a partial isomorphism 
between guarded subsets of $\frakA$ and $\frakA(T)$.
We claim that $\{ \gamma_v \circ \chi_v \mid v \in \nodes(T) \}$ 
is a guarded bisimulation between $\frakA$ and $\frakA(T)$. 

\noindent \emph{`Forth':} 
Consider $\gamma_v \circ \chi_v : A_v \to \{[v,c] \mid c\in K_v\}$ and $B$ 
a guarded subset of $\frakA$. Then, since $|B\cup A| \leq |K|=2n$, there is 
a vertex $\chi: B \to K$ such that $\chi_v|_{A_v\cap B} = \chi|_{A_v\cap B}$ 
and $\chi(A_v) \cap \chi'(B) = \chi(A_v\cap B)$. 
It follows that $\chi$ is adjacent to $\chi_v$ in $G$, 
hence $w = v \cdot \chi$ is adjacent to $v$ in $T$, $\pi(w) = \chi_w = \chi$, 
and that thus $\gamma_w \circ \chi_w$ fulfills the requirements 
of the `forth property'. 

\noindent \emph{`Back':}
Consider now $\gamma_v \circ \chi_v : A_v \to \{[v,c] \mid c\in K_v\}$ 
and a guarded subset $B = \{[w,d] \mid d \in D\}$ of $\frakA(T)$.
Let $C = D \cap K_v$. The intersection of $B$ and $\{[v,c] \mid c\in K_v\}$ 
consists of those $[v,c]$ such that $c \in C$ appears in the label 
of every node along the path $\rho$ connecting $v$ to $w$ in $T$.
Let $u$ and $y$ be adjacent nodes of $\rho$. 
Then $\pi(u)=\chi_u$ and $\pi(y)=\chi_y$ are adjacent in $G$ 
and thus $\chi_u^{-1}|C = \chi_y^{-1}|C$. 
By induction we get that $\chi_v^{-1}|C = \chi_w^{-1}|C$.
It follows that $\gamma_w \circ \chi_w$ satisfies the requirements 
of the `back property'.
\end{proof}

\noindent
This completes the proof of Proposition~\ref{prop_finite}, 
thereby also our Main Theorem~\ref{thrm_main}.




\bibliographystyle{model1-num-names}



\end{document}